\documentclass[12pt,usenames,dvipsnames]{article}

\usepackage{ragged2e}
\usepackage{sectsty}

\usepackage{amsmath,amsthm,amssymb,amscd,mathtools,bm,upgreek}
\usepackage{graphicx}
\usepackage{epsfig}
\usepackage{hyperref}
\hypersetup{colorlinks,linkcolor=blue,citecolor=blue,urlcolor = blue}
\usepackage{rotating}
\usepackage{float}
\usepackage{enumitem}
\usepackage{multirow}
\usepackage{booktabs}
\usepackage{eurosym}
\usepackage{caption}
\usepackage{pdfcolfoot}
\usepackage{longtable}
\usepackage[title]{appendix}
\RequirePackage{mathtools,algpseudocode,etoolbox}
\usepackage[linesnumbered,ruled,vlined]{algorithm2e}

\usepackage{listings}
\usepackage{color} 
\usepackage[style=apa,backend=biber]{biblatex}
\DeclareSourcemap{
  \maps[datatype=bibtex]{
    \map{
      \step[fieldset=issn, null]
      \step[fieldset=isbn, null]
      \step[fieldset=doi, null]
      \step[fieldset=url, null]
      \step[fieldset=note, null]
    }
  }
}
\usepackage{subcaption}
\usepackage{mwe}

\usepackage{siunitx}

\usepackage{xcolor,soul}
\definecolor{blueb}{RGB}{79,112,190}

\usepackage{authblk}

\def\comp{\raise 1pt \hbox{$\scriptstyle\circ$}}

\def\upto{{\raise 1pt \hbox{$\scriptstyle \,\nearrow\,$}}}
\def\downto{{\raise 1pt \hbox{$\scriptstyle \,\searrow\,$}}}

\newtheorem{proposition}{Proposition}

\newtheorem{assumption}{Assumption}
\newtheorem{paradox}{Paradox}

\title{\Large\textbf{Paradoxes of the public sector productivity measurement}
}

\author[1]{Timo Kuosmanen}
\author[2,\footnote{
Corresponding author. \newline \hspace*{5mm} 
\textit{E-mail addresses:} \texttt{timo.kuosmanen@utu.fi (T. Kuosmanen)}, \texttt{x.zhou@surrey.ac.uk (X. Zhou)}.}]{Xun Zhou}
\affil[1~]{Turku School of Economics, University of Turku, 20500 Turku, Finland}
\affil[2~]{{Surrey Business School, University of Surrey, Guildford GU2 7XH, UK}}
\date{September 2025}

\begin{document}

\maketitle

\vfill
\begin{center}
Declarations of interest: none
\end{center}
\vfill

\begin{abstract}
\noindent 
This paper critically investigates standard total factor productivity (TFP) measurement in the public sector, where output information is often incomplete or distorted. The analysis reveals fundamental paradoxes under three common output measurement conventions. When cost-based value added is used as the aggregate output, measured TFP may paradoxically decline as a result of genuine productivity-enhancing changes such as technical progress and improved allocative and scale efficiencies, as well as reductions in real input prices. We show that the same problems carry over to the situation where the aggregate output is constructed as the cost-share weighted index of outputs. In the case of distorted output prices, measured TFP may move independently of any productivity changes and instead reflect shifts in pricing mechanisms. Using empirical illustrations from the United Kingdom and Finland, we demonstrate that such distortions are not merely theoretical but are embedded in widely used public productivity statistics. We argue that public sector TFP measurement requires a shift away from cost-based aggregation of outputs and toward non-market valuation methods grounded in economic theory.
\\[5mm]
\\[2mm]
\end{abstract}
\vfill

\thispagestyle{empty}
\newpage
\setcounter{page}{1}
\setcounter{footnote}{0}
\pagenumbering{arabic}
\baselineskip 20pt
\setlength\bibitemsep{1.15\itemsep}

\section{Introduction}\label{sec:intro}
In competitive market sectors, productivity measurement is relatively straightforward, as a large variety of outputs can typically be valued (or weighted) using market prices, which reflect both demand-side valuations and supply-side costs. In the public sector, however, productivity measurement is challenging due to the absence of competitive pricing mechanisms, as market prices are often unavailable, distorted, or heavily regulated. The problem is more fundamental than merely a data availability issue: how to construct a meaningful measure of aggregate output when no reliable market-based valuation exists. In fact, this fundamental issue is reflected in the system of national accounts (SNA), which defaults to measuring public sector value added based on the cost of its inputs \parencite{united_nations_system_2009}. 

A useful starting point for discussing this measurement challenge is the framework provided by \textcite{diewert_productivity_2018}, which classifies the measures of public sector value added into three hierarchical categories based on the availability of quantity and price information. The first-best approach applies when both output quantities and market prices are available, assuming that market prices or consumer valuations provide a reliable aggregation of value added. The second-best approach applies when output quantities are observable but market prices are not, in which case unit production cost is used as a weight for each specific output to construct the value added. The third-best approach applies when neither output quantities nor market prices are available: output growth is proxied by real input growth, and output price changes are imputed using an index of input prices.
These approaches, which \textcite{diewert_productivity_2018} analyzes using index number methods, are also formulated differently: the first- and second-best are defined in terms of levels, whereas the third-best is defined in terms of changes.

In this paper, we focus on total factor productivity (TFP) measurement that relies on an aggregate production function in settings where output prices are unobservable or unreliable, or where both output prices and quantities are unavailable, as is often the case in public services. The choice of a production function framework is motivated by its generality. While it is an explicit assumption in econometric estimation (\cite{Olley1996}), the production function also serves as the implicit foundation for growth accounting \parencite{solow_technical_1957} and frontier-based approaches \parencite{Fare1994,Chambers1996productivity,dai_can_2025}.

The contribution of this paper is to formally demonstrate that standard approaches to TFP measurement can introduce systematic biases when applied in the absence of a reliable measure for aggregate output. Where aggregate output is constructed as cost-based value added or the cost-share weighted index of outputs, measured TFP may paradoxically fall as a direct result of genuine improvements such as technical progress, improved allocative efficiency, or improved scale efficiency. Changes in real input prices, regardless of whether they reflect underlying productivity changes, can similarly distort measured TFP. The resulting distortions can be so severe as to render the measurement exercise more misleading than informative. Moreover, when output prices are observable but distorted by regulation, measured TFP may merely reflect changes in regulatory pricing rather than changes in productivity.

This paper argues for a fundamental shift away from cost-based measurement of aggregate output in productivity analysis. The widespread yet often overlooked reliance on competitive market assumptions in TFP measurement has led to systematic biases, particularly in the context of the public sector. In settings where market valuations are unavailable or unreliable, we suggest using non-market valuation methods such as stated or revealed preferences to value and aggregate outputs more meaningfully.

The rest of the paper is organized as follows. Section \ref{sec:illustration} presents a motivating example. Section \ref{sec:setup} introduces the general theoretical framework for standard productivity analysis. Section \ref{sec:third-best} examines TFP measurement with cost-based value added and shows how this can produce counterintuitive TFP results. Section \ref{sec:third-best-2} extends the critique to TFP measurement with cost-weighted aggregate output, demonstrating that the same fundamental flaws carry over to this setting. Section \ref{sec:second-best} demonstrates that the measurement paradox can also arise from distorted output prices. Section \ref{sec:concl} concludes with implications for productivity analysis and directions for future research.

\section{A motivating example}\label{sec:illustration}
Before turning to the formal analysis, it is useful to consider a real-world example that illustrates the relevance of TFP measurement issues in the public sector. Figure \ref{fig:tfp-index} plots the TFP indices for education services and for human health and social work activities in the United Kingdom and Finland from 1995 to 2020, based on data from the EUKLEMS \& INTANProd database \parencite{bontadini_filippo_euklems_2023}.\footnote{Source: \url{https://euklems-intanprod-llee.luiss.it/}.} These indices are constructed using standard growth accounting methods that rely on cost-based value added as the output measure.

\begin{figure}[H]
        \centering
        \begin{subfigure}[b]{0.75\textwidth}
            \centering
            \includegraphics[width=\textwidth]{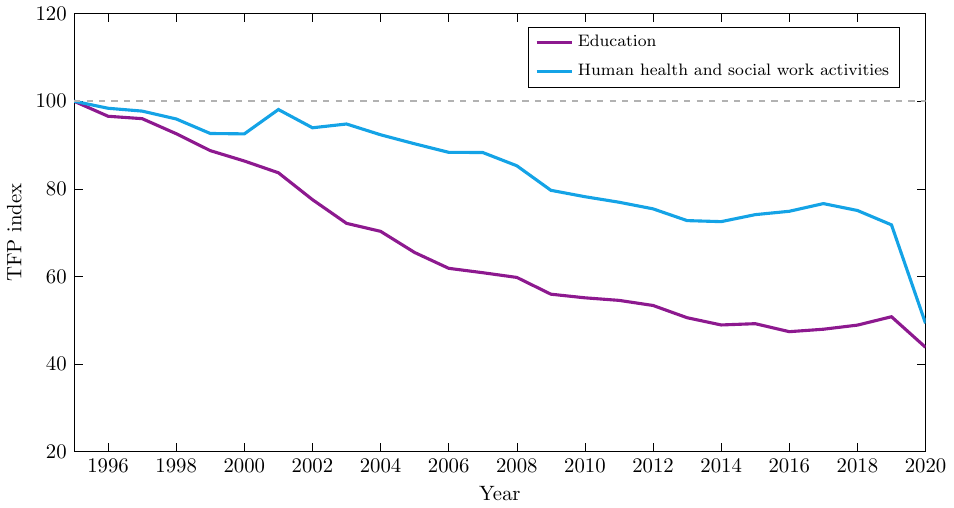}
                    \vspace{-1.2em}
            \caption[]%
            {{\small UK}}  
            \label{fig:uk}
        \end{subfigure}
        \hfill
        \begin{subfigure}[b]{0.75\textwidth}  
            \centering 
            \includegraphics[width=\textwidth]{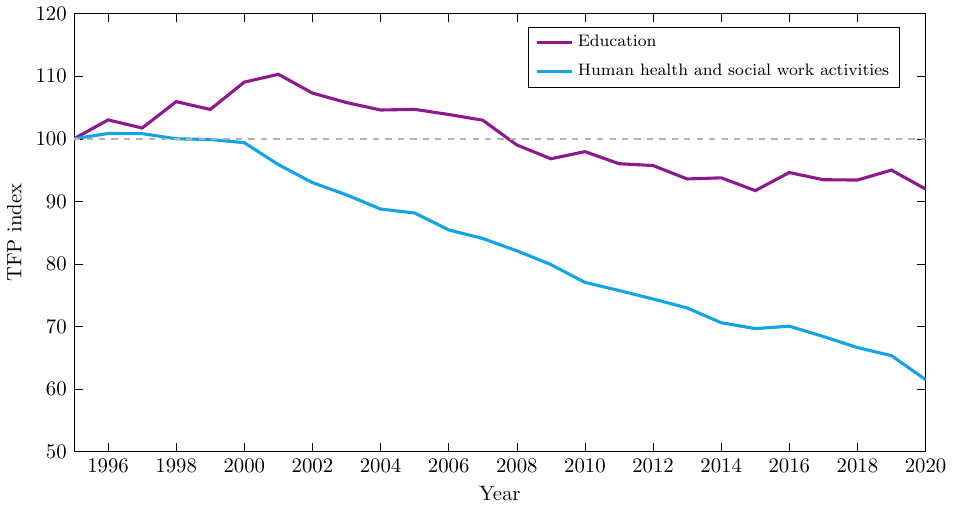}
                    \vspace{-1.2em}
            \caption[]%
            {{\small Finland}}    
            \label{fig:fi}
        \end{subfigure}
        \vspace{-0.2em}
        \caption[]
        {\small TFP indices for key public services in the UK and Finland; 1995--2020 (1995 = 100).} 
        \label{fig:tfp-index}
\end{figure}

In both countries, the measured TFP in education services and human health and social work activities has declined substantially over the past two decades. Note that Figure \ref{fig:tfp-index} depicts the level of TFP, not the growth rate. This observed decline of the standard growth accounting TFP does not conform with \textcite{diewert_productivity_2018}, who argues that TFP should remain constant in the third-best approach. The decreasing TFP presents an exceptionally severe case of Baumol's cost disease \parencite{baumol_performing_1966}, a theory predicting that productivity in the public sector tends not to increase. In the next sections, we will argue that the standard growth accounting TFP measures for public services such as education and healthcare presented in Figure \ref{fig:tfp-index} can be seriously misleading.

\section{Production theory}\label{sec:setup}
To formally examine the paradoxes in public sector TFP measurement, we begin with a standard neoclassical production model:
\begin{equation}\label{eq:solow}
    Y=A\cdot f(K,L),
\end{equation}
where $Y$ is the aggregate output measured by value added, $K$ is the capital input (the capital stock or the flow cost of capital services), $L$ is the labor input (number of employees or hours worked), $f$ is a production function, and $A$ is the Solow residual representing TFP \parencite{solow_technical_1957}. The model can be extended to include the gross output, intermediate inputs, or undesirable outputs (such as pollution) without altering the fundamental logic of our analysis. We do not restrict $f$ to any particular form. Note that possible estimation errors are beyond the scope of this paper.

Following \textcite{solow_technical_1957}, we can measure the TFP level as:
\begin{equation}\label{eq:TFP-solow}
    A = \frac{Y}{f(K, L)}.
\end{equation}
The TFP level captures the portion of aggregate output not attributed to input levels and serves as the central object of measurement in this paper. This is a standard approach to measuring TFP in competitive market sectors. It is worth noting that while economists are often interested in the growth rate of TFP, this paper focuses on the TFP level for analytical convenience.

Measurement of the aggregate output $Y$ is problematic in many public services, where output prices may be unavailable or unreliable, and in some cases, even output quantities are difficult to observe. We distinguish four categories ordered by increasing informational availability: 
\begin{enumerate}
    \item\textit{Cost-based measures of value added:} This category applies when neither output quantities nor prices are available, a situation common in public services where outputs are multidimensional, poorly defined, or difficult to quantify. In such cases, the analysis follows the convention of the SNA (income approach), where value added is constructed based on the cost of the inputs used by the public sector: 
    \begin{equation}\label{eq:third-best}
        \hat{Y} = rK + wL,
    \end{equation}
    where $r$ and $w$ are the real prices of capital and labor inputs, respectively. We denote the measured value added as $\hat{Y}$ to distinguish it from the true value added $Y$ defined in Equation \eqref{eq:solow}. 
    
    \item\textit{Cost-weighted aggregate output:} A related but distinct approach is used in cases where output quantities are measurable but market prices are unavailable. This is common in many public services due to the absence of market-based valuations. Examples of outputs include patients treated in healthcare, courses taught in education, cases resolved in law enforcement, and improved air or water quality in environmental protection. In this scenario, a measure of aggregate output, denoted by $\tilde{Y}$, is constructed as a cost-weighted output index:
    \begin{equation}\label{eq:second-best}
        \tilde{Y} = \sum_{i} c_i y_i,
    \end{equation}
    where $c_i$ represents the unit production cost assigned to each output $y_i$ in real terms. Although this approach uses output data, its valuation remains entirely derived from the input side.

    \item\textit{Value added based on distorted market prices:} The third category applies when output quantities and prices are observable, but the observed prices are distorted by regulation, subsidies, or other frictions. In these settings, the measured value added is directly affected by such distorted prices.

    \item\textit{Market-based value added:} The final category is the ideal scenario where both output quantities and their undistorted market prices are available. Here, value added can be reliably measured under conditions of perfect competition. As this rarely applies to the public sector and involves no theoretical bias, we do not further consider this case.
\end{enumerate}

Our classification builds on the conceptual typology in \textcite{diewert_productivity_2018} but differs in its structure. More importantly, our approach distinguishes itself from Diewert's in two key respects. First, we ground the analysis in a standard production function framework, rather than relying on index number methods. Second, we formulate the analysis entirely in terms of levels rather than growth rates. These distinctions allow us to formally examine how different constructions of aggregate output affect standard TFP measurement.

To operationalize our analysis, we draw on the efficiency concepts developed by \textcite{Farrell1957} and further extended by \textcite{Fare1994}. Specifically, we focus on three sources of productivity growth: improvements in technology, allocative efficiency, and scale efficiency. 

When technology improves, the production function $f$ shifts upward:
\begin{equation}\label{eq:tech-progress}
    f'(K, L) > f(K, L),\quad \forall K, L,
\end{equation}
where $f'$ denotes the post-shift technology. This upward shift (i.e., technical progress) raises the maximum attainable output for given inputs and hence contributes to higher TFP.

Allocative efficiency refers to whether inputs are allocated in a cost-minimizing way, given their relative prices. Improvements in allocative efficiency imply reallocation of inputs such that the maximum output can be obtained with reduced total input cost, thereby enabling a higher level of TFP. An input mix $(K, L)$ is allocatively efficient if the marginal rate of technical substitution equals the real input price ratio:
\begin{equation}\label{eq:allocative-eff}
    \frac{\partial f / \partial K}{\partial f / \partial L} = \frac{r}{w}.
\end{equation}

Scale efficiency concerns whether production occurs at the most productive scale size (MPSS). Under variable returns to scale, improving scale efficiency involves adjusting the scale of inputs toward the MPSS. For example, under increasing returns to scale, inputs can be scaled up by a factor $\lambda > 1$, leading to more than proportional growth in the maximum attainable output and a potential increase in TFP:
\begin{equation}\label{eq:scale-eff}
    f(\lambda K, \lambda L) > \lambda f(K, L).
\end{equation}
Under decreasing returns to scale, the opposite applies: improving scale efficiency requires a contraction of inputs, resulting in a less than proportional decrease in the maximum attainable output and a potential increase in TFP.

\section{TFP measurement with cost-based value added}\label{sec:third-best} 
In the case of no output quantities or prices, as discussed in Section \ref{sec:setup}, the value added of public service production is measured based on the total input cost as in Equation \eqref{eq:third-best}. The corresponding TFP level, $\hat{A}$, is measured as:
\begin{equation}\label{eq:TFP-no-info}
    \hat{A} = \frac{\hat{Y}}{f(K, L)} = \frac{rK + wL}{f(K, L)}.
\end{equation}

The fundamental problem with the TFP measurement defined in Equation \eqref{eq:TFP-no-info} is its complete lack of connection to outputs, as both the numerator (total input cost) and the denominator (the production function) are functions of the same inputs, $K$ and $L$. 
This construction, which depends entirely on input data, makes it straightforward to show the following counterintuitive findings.

\begin{paradox}\label{theo:paradox1}
Under constant input quantities and real input prices, technical progress reduces the TFP measured according to Equation \eqref{eq:TFP-no-info}.
\end{paradox}

\begin{proof}
Technical progress is represented by an upward shift of the production function as in Equation \eqref{eq:tech-progress}. Since $r, w, K, L$ are assumed to remain unchanged, the measured TFP in Equation \eqref{eq:TFP-no-info} becomes:
\begin{equation}
    \hat{A}' = \frac{rK + wL}{f'(K,L)} < \frac{rK + wL}{f(K,L)} = \hat{A}.
\end{equation}
\end{proof}

Paradox \ref{theo:paradox1} may arise in a variety of public services. For instance, a public clinic may adopt a free artificial intelligence (AI)-based triage chatbot to help screen patient inquiries \parencite{jindal_ensuring_2024}, allowing staff to process more cases per day without increasing labor or capital inputs. Similarly, teachers in public schools may begin using free generative AI tools to automate feedback on student assignments \parencite{chiu_impact_2024}, thereby raising instructional output without altering staffing levels or institutional costs. In these situations, the TFP measured according to Equation \eqref{eq:TFP-no-info} may paradoxically decline when technical progress fosters genuine improvements in productivity.

\begin{paradox}\label{theo:paradox2}
Under constant technology and real input prices, improvements in allocative efficiency reduce the TFP measured according to Equation \eqref{eq:TFP-no-info}.
\end{paradox}

\begin{proof}
Consider an initial input mix $(K, L)$ that can obtain the maximum output $f(K, L)$ at cost $C = rK + wL$. Improving allocative efficiency implies reallocation of inputs to a new mix $(K', L')$ that can produce the same maximum attainable output at a lower cost. In other words, we have: 
\begin{equation}
f(K', L') = f(K, L), \quad C' = rK' + wL' < C.
\end{equation}
Thus, the measured TFP in Equation \eqref{eq:TFP-no-info} becomes:
\begin{equation}
\hat{A}' = \frac{C'}{f(K', L')} = \frac{C'}{f(K, L)} < \frac{C}{f(K, L)} = \hat{A}.
\end{equation}
\end{proof}

Paradox \ref{theo:paradox2} may be observed in practice when a public organization reallocates inputs between capital and labor to reduce cost for a given output level. For example, a hospital may shift administrative work from salaried clerical staff (labor) to an existing self-service kiosk system (capital) without changing the scale or quality of services delivered. In education, schools may reduce reliance on teaching assistants (labor) by increasing scheduled use of existing instructional software or resources (capital), while keeping the class hours fixed. Similar reallocations can occur within input categories. In a hospital, administrative duties may be shifted from higher-paid doctors to lower-paid assistants; in education, routine tasks may be delegated from senior to junior teaching staff. In these situations, the TFP measured according to Equation \eqref{eq:TFP-no-info} may paradoxically decline when improvements in allocative efficiency allow for productivity growth.

\begin{paradox}\label{theo:paradox3}
Under constant technology and real input prices, improvements in scale efficiency reduce the TFP measured according to Equation \eqref{eq:TFP-no-info}.
\end{paradox}

\begin{proof}
Consider an initial input mix $(K, L)$ that can obtain the maximum output $f(K, L)$ at cost $C = rK + wL$. To improve scale efficiency under increasing returns to scale, inputs are scaled up by a factor $\lambda > 1$ toward the MPSS. In this case, the maximum attainable output increases more than proportionally, as in Equation \eqref{eq:scale-eff}, while the new cost is: 
\begin{equation}
    C' = \lambda(rK + wL).
\end{equation}
Thus, the measured TFP in Equation \eqref{eq:TFP-no-info} becomes:
\begin{equation}
\hat{A}' = \frac{\lambda(rK + wL)}{f(\lambda K, \lambda L)} < \frac{\lambda(rK + wL)}{\lambda f(K, L)} = \frac{rK + wL}{f(K, L)} = \hat{A}.
\end{equation}

Analogously, to improve scale efficiency under decreasing returns to scale, inputs are scaled down by a factor $\delta < 1$ toward the MPSS. In this case, the maximum attainable output decreases less than proportionally: 
\begin{equation}
    f(\delta K, \delta L) > \delta f(K, L), \quad C' = \delta(rK + wL).
\end{equation}
Thus, the measured TFP in Equation \eqref{eq:TFP-no-info} becomes:
\begin{equation}
\hat{A}' = \frac{\delta(rK + wL)}{f(\delta K, \delta L)} < \frac{\delta(rK + wL)}{\delta f(K, L)} = \frac{rK + wL}{f(K, L)} = \hat{A}.
\end{equation}
\end{proof}

Paradox \ref{theo:paradox3} may arise when public services adjust input scale toward the MPSS without changing technology or real input prices. For example, a public transit agency may increase route frequency by proportionally expanding drivers and vehicles, or a public school system may respond to growing enrollment by hiring additional teachers and opening more classes. If the sector operates under increasing returns to scale, the maximum attainable output increases more than proportionally, yet the TFP measured according to Equation \eqref{eq:TFP-no-info} may paradoxically decline despite the potential productivity growth.

\begin{paradox}\label{theo:paradox4}
Under constant technology and input quantities, lower real input prices reduce the TFP measured according to Equation \eqref{eq:TFP-no-info}.
\end{paradox}

\begin{proof}
Suppose real input prices decrease with $r' < r$ and $w' < w$, while input quantities and the production function remain fixed. The measured TFP in Equation \eqref{eq:TFP-no-info} becomes:
\begin{equation}
    \hat{A}' = \frac{r'K + w'L}{f(K,L)} < \frac{rK + wL}{f(K,L)} = \hat{A}.
\end{equation}
\end{proof}

Real-world examples of Paradox \ref{theo:paradox4} include cases where public sector entities reduce unit input prices without changing input quantities or technology. For example, a municipality may renegotiate procurement contracts for utilities or medical supplies at lower prices, or a school system may adopt centralized purchasing to reduce the unit costs of materials. Alternatively, public service wages may decline as a result of collective bargaining or public sector wage policy reforms. Although these changes lower real input costs, they do not reflect any underlying change in productivity. Nevertheless, the TFP measured according to Equation \eqref{eq:TFP-no-info} may paradoxically decline simply because lower costs are interpreted as a reduction in the measured value added.

Overall, these four paradoxes reveal fundamental flaws in standard TFP measurement under the cost-based approach, as is often the case in the public sector. Because the measure reflects only input costs rather than actual productivity, it systematically misinterprets genuine improvements (such as technical progress and improvements in allocative efficiency and scale efficiency) as declines in the measured TFP. Moreover, reductions in real input prices, regardless of whether they reflect underlying productivity changes, can produce similar measurement distortions. Therefore, any application of this approach may need to be interpreted with ``reverse psychology": a decline in the measured TFP could in fact signal either an underlying improvement in productivity or simply a reduction in real input prices. Fundamentally, this cost-based approach assigns value to outputs based on input requirements, a logic closely aligned with the labor theory of value, which is central to the Marxist approach to economic valuation.

\section{TFP measurement with cost-weighted aggregate output}\label{sec:third-best-2} 

One might assume the paradoxes presented in Section \ref{sec:third-best} only arise in the extreme case of no output information. However, the same fundamental flaws persist even when output quantities are observable, as long as the valuation of these outputs remains derived from input costs. This is precisely the case in many public services, where output quantities are observable but output prices are unavailable or effectively zero. Examples include public library services (measured by book loans or visitor numbers), maintenance of public parks (measured by maintained area), general policing or fire services (measured by incident response or coverage), and the core outputs of public universities such as the number of courses taught, degrees awarded, or research articles published. In this case, the aggregate output is commonly constructed as a unit production cost-weighted output index as in Equation \eqref{eq:second-best}. 

To formalize the key result of this approach, we first state a critical assumption regarding cost allocation.

\begin{assumption}\label{theo:assump1}
    The total real input cost, $rK + wL$, can be exhaustively and proportionally attributed to a set of observed output indicators $y_i$ according to cost shares $\alpha_i > 0$ where $\sum_i \alpha_i = 1$.
\end{assumption}

This assumption, which is standard in the cost-weighting approach, leads to the following proposition.

\begin{proposition}\label{theo:propo1}
    Under Assumption \ref{theo:assump1}, the cost-weighted aggregate output index $\tilde{Y}$ is identical to the total real input cost.
\end{proposition}

\begin{proof}
    Under Assumption \ref{theo:assump1}, the unit cost $c_i$ for each output $y_i$ is defined as:
    \begin{equation}\label{eq:unit-cost}
        c_i = \alpha_i (rK + wL)/y_i, \quad \sum_i \alpha_i = 1,
    \end{equation}
Substituting this definition into the cost-weighted output index from Equation \eqref{eq:second-best} yields:
    \begin{equation}\label{eq:third-equiv-second}
        \tilde{Y} = \sum_i c_i y_i = \sum_i \alpha_i (rK + wL) = rK + wL.
    \end{equation}
\end{proof}

Proposition \ref{theo:propo1} reveals that, despite having information on output quantities, the resulting aggregate output construction remains entirely input-driven. The measured TFP under this construction is identical to the case where neither output prices nor quantities are observed. As a result, the biases documented in Paradoxes  \ref{theo:paradox1}--\ref{theo:paradox4} fully carry over to this setting. 

The result in Proposition \ref{theo:propo1} depends on the assumption of exhaustive cost allocation. In practice, for institutions such as hospitals or schools, this is rarely the case. There are often unmeasured activities, and a significant portion of costs (e.g., for general administration and management) can be considered ``overhead" that is difficult to allocate to specific outputs. However, even in a more realistic scenario where the cost-weighted index covers only a part of the total cost, our central critique remains valid. 

It is also useful to note the level of aggregation at which these different measures are applied. The SNA-based value added discussed in Section 4 is most commonly used in industry-level productivity analysis. In contrast, the cost share-weighted sum of output indicators, as analyzed in this section, is more standard at more disaggregated levels, such as for individual establishments (e.g., hospitals or schools).

\section{TFP measurement with distorted market prices}\label{sec:second-best}
In some public services, both output quantities and prices are observable, but the prices are heavily regulated or distorted. Examples include public transportation systems with fixed ticket prices, public utility providers with administratively set tariffs, and national postal services with regulated stamp prices. In this setting, the most direct monetary measure for aggregate output is total revenue. The corresponding TFP level, $\hat{A}$, is therefore measured as:
\begin{equation}\label{eq:TFP-partial-info-2}
\hat{A} = \frac{R}{f(K, L, M)}.
\end{equation}
where $R$ is the total revenue, $M$ represents intermediate inputs, and $f$ is now the gross output production function. This leads to the following paradox:

\begin{paradox}
Under constant technology, real input prices, and input and output quantities, a decrease in regulated output prices reduces the TFP measured according to Equation \eqref{eq:TFP-partial-info-2}.
\end{paradox}

\begin{proof}
Consider a public service provider where the price for each unit of service is set as a markup over marginal cost:
\begin{equation}
p_i = (1+\mu_i) MC_i,
\end{equation}
where $\mu_i$ represents the markup and $MC_i$ is the real marginal cost for output $i$. The provider's total revenue is then given by:
\begin{equation}
R = \sum_i p_i y_i = \sum_i (1+\mu_i) MC_i y_i.
\end{equation}

Suppose a regulatory intervention reduces the markup $\mu_i' < \mu_i$ for all $i$ (e.g., through tighter price controls). Since it is assumed that technology and real input prices are constant, the marginal cost $MC_i$ for each output $i$ remains unchanged. With output quantities also held constant, the measured TFP in Equation \eqref{eq:TFP-partial-info-2} becomes:
\begin{equation}
\hat{A}' = \frac{\sum_i (1+\mu_i') MC_i y_i}{f(K, L, M)} < \frac{\sum_i (1+\mu_i) MC_i y_i}{f(K, L, M)} = \hat{A}.
\end{equation}
\end{proof}

In sum, the analysis of distorted prices reveals that fundamental flaws in standard TFP measurement persist even when moving from cost-based measures of aggregate output to monetary output aggregates. The resulting TFP measure incorrectly attributes regulatory changes to productivity changes, leading to misleading conclusions about public sector performance.

It is worth noting that, while value added is often the final measure of interest, the analytical advantage of using total revenue in this paper lies in its direct functional relationship with the distorted prices. This allows us to most clearly isolate the paradox that arises from price regulation. Consistent with using total revenue (i.e., gross output) in the numerator, the denominator of the TFP formula \eqref{eq:TFP-partial-info-2} features a gross output production function that includes intermediate inputs in addition to capital and labor. The proper identification of such gross output production functions is a topic of ongoing research (see, e.g., \cite{gandhi_identification_2020}). While our simplified framework effectively illustrates the paradox arising from price distortions, the choice between the value added and gross output approaches for public sector productivity analysis certainly warrants further research.

\section{Concluding remarks}\label{sec:concl}

This paper has shown that standard approaches to TFP measurement can yield systematically biased results when applied in the absence of a reliable measure of aggregate output. In common settings in the public sector where value added is defined based on input aggregates or the aggregate output is constructed using cost-based weights, measured TFP may paradoxically decline as a result of genuine productivity-enhancing changes, such as technical progress, improved allocative efficiency, or improved scale efficiency. It may also simply respond to changes in real input prices without reflecting any underlying change in productivity. Furthermore, even when monetary revenue aggregates are available but based on distorted prices, the TFP measure remains biased, as it incorrectly attributes regulatory output price changes to productivity changes.

These findings demonstrate the importance of revisiting how productivity is assessed in the public sector. While perfect measures may not exist, it is crucial to recognize that not all productivity measures are equally flawed, and indeed, bad measures can create the wrong kinds of incentives. Policymakers and statistical agencies should exercise extreme caution when interpreting TFP indicators based on cost-based measures of aggregate output or monetary output aggregates, particularly in settings where market-based valuation is infeasible or distorted. If such proxies are used instead of meaningful output valuation, productivity indices may mislead decision makers, generating counterproductive incentives that distort resource allocation. For example, these measures may inadvertently incentivize organizations to raise wages and capital compensations, thereby inflating ``measured productivity" without genuine service improvement. Moreover, if outputs are weighted by their associated unit costs (i.e., Marxist labor theory of value), organizations may prioritize producing the most expensive services possible, even if simpler, more beneficial alternatives exist. This highlights the critical importance of selecting appropriate output valuation methods in public services to avoid unintentionally encouraging inefficiency or unnecessary expenditure, and ultimately to ensure that incentives align with the public good.

A central implication is that public sector productivity analysis must shift away from reliance on cost-based proxies, and instead move toward independent and economically grounded valuations of output. When market prices are unavailable or unreliable, non-market valuation methods, including stated preferences and revealed preferences techniques (e.g., \cite{Kuosmanen2021,chay_does_2005,hainmueller_consumer_2015}), offer promising tools for valuing public service outputs. These methods are already well established in fields such as environmental economics, which similarly deals with valuing services that lack clear market prices. A more systematic application of these methods to productivity measurement in health, education, and other public services represents a fascinating direction for future research.

Finally, the measurement paradoxes identified in this paper are not merely historical artifacts of an outdated accounting framework. The recently adopted 2025 System of National Accounts continues to rely on the same cost-based valuation for non-market output \parencite{united_nations_system_2025}. This underscores the persistent and forward-looking relevance of our findings, and it highlights the urgent need to develop and implement the alternative, economically-grounded valuation methods advocated in this paper.



\printbibliography
\baselineskip 12pt





\end{document}